\documentclass[10pt]{amsart}
\usepackage[all]{xy}
\usepackage{amsmath}
\usepackage{amsfonts}
\usepackage{amssymb}
\usepackage{amscd}
\usepackage{amsthm}
\usepackage{latexsym}
\usepackage{amsbsy}
\usepackage{graphicx}
\usepackage{epstopdf}
\AppendGraphicsExtensions{.gif}
\newtheorem{lemma}{Lemma}[section]
\newtheorem{proposition}{Proposition}[section]
\newtheorem{theorem}{Theorem}[section]

\usepackage{lipsum}
\title{Rationality of Spectral Action for Robertson-Walker Metrics}

\author{ $ $\\ Farzad Fathizadeh, Asghar Ghorbanpour, Masoud Khalkhali}

\begin{document}

\maketitle

\begin{center} 

Department of Mathematics, Western University \\  
London, Ontario, Canada, N6A 5B7 \footnote{{\it E-mail addresses}:  
ffathiz@uwo.ca, aghorba@uwo.ca, masoud@uwo.ca}
 \\

\end{center}

\begin{abstract}

We use pseudodifferential calculus and heat kernel techniques to 
prove a conjecture by Chamseddine and Connes on rationality of the coefficients 
of the polynomials in the cosmic scale factor $a(t)$ and its higher derivatives, which 
describe the general terms $a_{2n}$ in the expansion of the spectral action for general 
Robertson-Walker metrics. We also compute the 
terms up to  $a_{12}$  in the expansion of the spectral action by our method. As a byproduct, we verify that    our 
computations agree with the terms up to $a_{10}$  that were 
previously computed by Chamseddine and Connes by a different method.

\end{abstract}

\vskip 0.1cm


\vskip 0.1 cm


\tableofcontents

\section{Introduction}

Noncommutative geometry in the sense of Alain Connes \cite{ConBook} has provided a 
paradigm for geometry in the noncommutative setting based 
on spectral data. This generalizes Riemannian geometry \cite{ConReconstruct} and incorporates 
physical models of elementary particle physics 
\cite{ConGravity, ConMixing, ChaConMarGS, ConMarBook, ChaConConceptual, ChaConWhy, GraIocSch, Sit, Sui1, Sui2}. 
 An outstanding  feature of the spectral action defined for  
noncommutative geometries  
is that it derives the Lagrangian of the physical models from 
simple noncommutative geometric data \cite{ConMixing, ChaConSAP, ChaConMarGS}. 
Thus various methods have been developed 
for computing the terms in the expansion in the energy scale $\Lambda$ of the spectral action \cite{ChaConUFNCG, ChaConGravity, ChaConUncanny, ChaConRW, IocLevVasGlobal, IocLevVasTorsion}. 
Potential applications of noncommutative geometry in cosmology  have recently been 
carried out in \cite{KolMar, Mar, MarPie, MarPieTeh2012, MarPieTeh, NelOchSal, NelSak1, NelSak2, EstMar}.

Noncommutative geometric spaces are described by spectral triples $(\mathcal{A}, \mathcal{H}, D)$, 
where $\mathcal{A}$ is an involutive algebra represented by bounded operators on a Hilbert space 
$\mathcal{H}$, and $D$ is an unbounded self-adjoint operator acting in $\mathcal{H}$ \cite{ConBook}. 
The operator $D$, which plays the role of the Dirac operator,  encodes the metric information and it is further 
assumed that it has bounded commutators with elements of $\mathcal{A}$. It has been shown that 
if $\mathcal{A}$ is commutative and the triple satisfies suitable regularity conditions then $\mathcal{A}$ 
is the algebra of smooth functions on a spin$^c$ manifold $M$ and $D$ is the Dirac operator acting in 
the Hilbert space of $L^2$-spinors \cite{ConReconstruct}.  In this case, the Seeley-de Witt coefficients $a_{n}(D^2) = \int_M a_n (x, D^2) \,dv(x)$, which vanish for odd  $n$,  
appear in a small time asymptotic expansion of the form
\[
\textnormal{Tr}(e^{-t D^2}) \sim t^{- \textnormal{dim} (M)/2} \sum_{n\geq 0} a_{2n} (D^2) t^n  \qquad (t \to 0). 
\] 
These coefficients determine the terms in the expansion of the spectral action. That is, there is an expansion of the form
\[
\textnormal{Tr} f(D^2/\Lambda^2) \sim \sum_{n \geq 0} f_{2n}\, a_{2n} (D^2/\Lambda^2), 
\]
where $f$ is a positive even function defined on the real line, and  $f_{2n} $ are the moments of 
the function $f$ \cite{ChaConSAP, ChaConUFNCG}. See Theorem 1.145 in \cite{ConMarBook} for  
details in a more general setup, namely for spectral triples with simple dimension spectrum.

By devising a direct method based on the Euler-Maclaurin formula and the 
Feynman-Kac formula, Chamseddine and Connes have initiated in 
\cite{ChaConRW} a detailed study of the spectral action for the Robertson-Walker 
metric with a general cosmic scale factor $a(t)$. They calculated the terms up to $a_{10}$ in 
the expansion and checked the agreement of the terms up to  $a_6$ against Gilkey's 
universal formulas \cite{GilBook1, GilBook2}.

The present paper is intended  to compute the term $a_{12}$ in the spectral action 
for general Robertson-Walker metrics, and to prove the conjecture of Chamseddine and 
Connes \cite{ChaConRW} on rationality of the coefficients of the polynomials in $a(t)$ and its derivatives 
that  describe the general terms $a_{2n}$ in the expansion. 
In passing, we compare the outcome of our computations up to the term $a_{10}$ with the 
expressions obtained in \cite{ChaConRW}, and confirm their agreement.

 In terms of the above aims, explicit formulas for the Dirac operator of the Robertson-Walker metric and 
its pseudodifferential symbol in Hopf coordinates are derived in \S \ref{DiracinHopf}. Following a brief review 
of the heat kernel method for computing local invariants of elliptic differential operators using pseudodifferential 
calculus \cite{GilBook1}, we compute in \S \ref{Termsupto10} the terms up to $a_{10}$ in the expansion of the spectral action for Robertson-Walker 
metrics. The outcome of our calculations confirms the expressions obtained in \cite{ChaConRW}. This forms a check 
in particular on the validity of $a_8$ and $a_{10}$, which as suggested in \cite{ChaConRW} also, seems necessary due to the high complexity of the formulas.  In \S \ref{Term12}, we record the expression 
for the term $a_{12}$ achieved by a significantly heavier computation, compared to the previous terms. It is checked   that 
the reduction of $a_{12}$ to the round case $a(t)=\sin t $ conforms to the full expansion obtained 
in \cite{ChaConRW} for the round metric  by remarkable calculations that are based on the Euler-Maclaurin formula. 
In order to validate our expression for $a_{12}$,  parallel but completely different computations are performed in 
spherical coordinates and the final results are confirmed to match precisely with our calculations in Hopf coordinates.

In \S \ref{ProofofConjecture}, we prove the conjecture made in \cite{ChaConRW} on rationality of the coefficients appearing in the expressions for the terms of the spectral action for Robertson-Walker metrics.  That is, we show that 
the term $a_{2n}$ in the expansion is of the form $Q_{2n}\big(a(t),a'(t),\dots,a^{(2n)}(t)\big)/a(t)^{2n-3}$, 
where $Q_{2n}$ is a polynomial with rational coefficients. We also find a formula for the coefficient of the 
term with the highest derivate of $a(t)$ in $a_{2n}$. 
It is known that  values of Feynman integrals  for quantum  gauge theories are closely  related to  multiple zeta values and periods in general and hence tend to be  transcendental  numbers \cite{MarBook}. 
In sharp distinction, the rationality result proved in this paper is valid for all scale factors $a(t)$  in  Robertson-Walker metrics. Although it might be exceedingly difficult, it is certainly desirable to find all the terms $a_{2n}$ in the spectral action. The rationality result is a consequence of a certain symmetry in the heat kernel and it is plausible that this symmetry would eventually reveal the full structure of the coefficients $a_{2n}$.  This is a task for a future work. 
Our main conclusions are summarized in \S \ref{Conclusions}.

\section{The Dirac Operator for Robertson-Walker Metrics}\label{DiracinHopf}

According to the spectral action principle \cite{ConGravity, ChaConSAP}, the spectral action of any geometry depends on its Dirac operator 
since the terms in the expansion are determined by the high frequency 
behavior of the eigenvalues of this operator.  For spin manifolds,  
the explicit computation of the Dirac operator in a coordinate system 
is most efficiently 
achieved by writing its formula after lifting the Levi-Civita connection on 
the cotangent bundle to the spin connection on the spin bundle. In this section, we 
summarize this formalism  and compute 
the Dirac operator of the Robertson-Walker metric  in Hopf coordinates. Throughout this paper we use Einstein's summation convention without any further notice.

\subsection{Levi-Civita connection.}

The spin connection of any spin manifold $M$ is the lift of the 
Levi-Civita connection for the cotangent bundle $T^*M$ to the 
spin bundle. Let us, therefore, recall the following recipe for 
computing the Levi-Civita connection and thereby the spin 
connection of $M$. Given an orthonormal frame $\{\theta_\alpha\}$ for 
the tangent bundle $TM$ and its dual coframe $\{\theta^\alpha\}$, the 
connection 1-forms  $\omega^\alpha_\beta$ of any connection $\nabla$ on $T^*M$  
are defined by
\begin{equation} \nonumber
\nabla{\theta^\alpha}=\omega_\beta^\alpha \,\theta^\beta.
\end{equation}

Since the Levi-Civita connection is the unique torsion free connection 
which is compatible with the metric, its 1-forms are uniquely determined by 
\[
d\theta^\beta =\omega^\beta_\alpha \wedge \theta^\alpha. 
\]
This is justified by the fact that the compatibility with metric enforces the 
relations 
\begin{equation}\nonumber
\omega^\alpha_\beta=-\omega^\beta_\alpha, 
\end{equation}
while, taking advantage of the first Cartan structure equation, 
the torsion-freeness amounts to the vanishing of 
\begin{equation}\nonumber
T^\alpha = d\theta^\alpha - \omega^\alpha_\beta \wedge \theta^\beta.
\end{equation}

\subsection{The spin connection of Robertson-Walker metrics in Hopf coordinates. }

The (Euclidean) Robertson-Walker metric  with the cosmic scale factor $a(t)$ is given by  
\begin{equation} \nonumber
ds^{2}=dt^{2}+a^{2}\left(  t\right)   d\sigma^2,  
\end{equation} 
where $d\sigma^2$ is the round metric on the 3-sphere $\mathbb{S}^3$. It is customary to write this metric in spherical coordinates,  
however, for our purposes which will 
be explained below, it is more convenient  to use the Hopf coordinates, 
which parametrize the 3-sphere   $S^3\subset \mathbb{C}^2$ by
\begin{equation} \nonumber
z_1=e^{i\phi_1}\sin(\eta), \qquad 
z_2=e^{i\phi_2}\cos(\eta),
\end{equation}
with $\eta$ ranging in  $[0,\pi/2)$ and $\phi_1,\phi_2$ ranging in $ [0,2\pi)$.  The 
Robertson-Walker metric in the coordinate system $x=(t, \eta, \phi_1, \phi_2)$  is 
thus given by
\begin{equation} \nonumber
ds^{2}=dt^{2}+a^{2}\left(  t\right)\left(d\eta^2+\sin^2(\eta)d\phi_1^2+\cos^2(\eta)d\phi_2^2\right). 
\end{equation}
An orthonormal coframe for $ds^{2}$ is then provided by 
 \begin{eqnarray}
 \theta^1= dt, \qquad 
 \theta^2 =a(t)\, d \eta, \qquad 
 \theta^3 = a(t)\, \sin \eta \,d \phi_1, \qquad 
 \theta^4 = a(t)\, \cos \eta \, d \phi_2. \nonumber 
\end{eqnarray}
Applying the exterior derivative to these forms, one can easily show that they satisfy the 
following equations, which determine the connection 1-forms of the Levi-Civita connection: 
\begin{eqnarray}
&& d\theta^1= 0, \nonumber \\
&&  d\theta^2 =\frac{a'(t)}{a(t)}\, \theta^1\wedge \theta^2, \nonumber \\ 
&& d\theta^3 =\frac{a'(t)}{a(t)}\, \theta^1\wedge \theta^3+ \frac{\cot\eta}{a(t)}\, \theta^2\wedge \theta^3,  \nonumber \\
 && d\theta^4 =\frac{a'(t)}{a(t)}\, \theta^1\wedge \theta^4- \frac{\tan\eta}{a(t)}\, \theta^2\wedge \theta^4.  \nonumber 
\end{eqnarray}
We recast the above equations into the matrix of connection 1-forms  
\[\omega=\frac{1}{a(t)}\left(
\begin{array}{cccc}
 0 & - a'(t)\, \theta ^2 & - a'(t)\, \theta ^3 & - a'(t)\, \theta ^4 \\
  a'(t)\, \theta ^2 & 0 & -\cot \eta \,  \theta ^3 &  \tan \eta \, \theta ^4 \\
  a'(t)\, \theta ^3 & \cot \eta  \, \theta ^3 & 0 & 0 \\
  a'(t)\, \theta ^4 & - \tan\eta \, \theta ^4 & 0 & 0 \\
\end{array} 
\right) \in\mathfrak{so}(4), \]
which lifts to the spin bundle using the Lie algebra isomorphism 
$\mu:\mathfrak{so}(4)\to \mathfrak{spin}(4)$ given by 
(see \cite{LawMic})
\begin{equation} \nonumber
\mu(A)= \frac{1}{4}\sum_{\alpha,\beta}\langle A\theta^\alpha,\theta^\beta\rangle c(\theta^\alpha)c(\theta^\beta), \qquad A \in \mathfrak{so}(4).
\end{equation}
Since $\langle \omega\theta^\alpha,\theta^\beta\rangle =\omega^\alpha_\beta$, 
the lifted connection $\tilde{\omega}$ is written as 
\[
\tilde\omega=\frac{1}{4}\sum_{\alpha,\beta}\omega^\alpha_\beta c(\theta^\alpha)c(\theta^\beta). 
\]
In the case of the Robertson-Walker metric we find that  
\begin{equation} \label{exprspinconn}
\tilde\omega=\frac{1}{2a(t)}
\left( a'(t)\theta ^2 \gamma ^{12}+
a'(t)\theta ^3 \gamma ^{13}+
a'(t) \theta ^4\gamma ^{14}+
   \cot (\eta )  \theta ^3\gamma ^{23}-
   \tan (\eta )  \theta ^4\gamma ^{24}\right), 
\end{equation}
where we use the notation $\gamma^{i j} = \gamma^i \gamma^j$ for products of pairs of the gamma matrices 
$\gamma^1, \gamma^2, \gamma^3, \gamma^4$, which are respectively written as
\begin{small}
\[
\left(
\begin{array}{cccc}
 0 & 0 & i & 0 \\
 0 & 0 & 0 & i \\
 i & 0 & 0 & 0 \\
 0 & i & 0 & 0
\end{array}
\right), 
 \left(
\begin{array}{cccc}
 0 & 0 & 0 & 1 \\
 0 & 0 & 1 & 0 \\
 0 & -1 & 0 & 0 \\
 -1 & 0 & 0 & 0
\end{array}
\right),  
\left(
\begin{array}{cccc}
 0 & 0 & 0 & -i \\
 0 & 0 & i & 0 \\
 0 & i & 0 & 0 \\
 -i & 0 & 0 & 0
\end{array}
\right),
\left(
\begin{array}{cccc}
 0 & 0 & 1 & 0 \\
 0 & 0 & 0 & -1 \\
 -1 & 0 & 0 & 0 \\
 0 & 1 & 0 & 0
\end{array}
\right).
\]
\end{small}

\subsection{The Dirac Operator of Robertson-Walker metrics in Hopf coordinates.}

Using the expression  \eqref{exprspinconn} obtained for the spin connection and considering the predual of the 
orthonormal coframe $\{ \theta^\alpha \}$, 
\begin{eqnarray}
 \theta_1= \frac{\partial}{\partial t}, \qquad 
 \theta_2=\frac{1}{a(t)}\frac{\partial}{\partial \eta}, \qquad
\theta_3 = \frac{1}{a(t)\, \sin \eta} \frac{\partial}{\partial \phi_1}, \qquad
 \theta_4 = \frac{1}{a(t)\, \cos \eta} \frac{\partial}{\partial \phi_2}, \nonumber 
\end{eqnarray}
we compute the Dirac operator for the Robertson-Walker metric explicitly: 
\begin{align*}
D&=c(\theta^\alpha)\nabla_{\theta_\alpha} \\
&=\gamma^\alpha\left(\theta_\alpha+\tilde \omega(\theta_\alpha)\right)\\
&=\gamma^1\left(\frac{\partial}{\partial t}\right)+
\gamma^2\left(\frac{1}{a}\frac{\partial}{\partial \eta}+\frac{a'}{2a}\gamma^{12}\right)
+\gamma^3\left(\frac{1}{a\sin(\eta)}\frac{\partial}{\partial \phi_1}+\frac{a'}{2a}\gamma^{13}+\frac{\cot(\eta)}{2a}\gamma^{23}\right)\\
& \quad +\gamma^4\left(\frac{1}{a\cos(\eta)}\frac{\partial}{\partial \phi_2}+\frac{a'}{2a}\gamma^{14}-\frac{\tan(\eta)}{2a}\gamma^{24}\right)\\
&=\gamma^1 \frac{\partial}{\partial t}+\gamma^2 \frac{1}{a}\frac{\partial}{\partial \eta}+\gamma^3 \frac{1}{a\, \sin \eta} \frac{\partial}{\partial \phi_1}+\gamma^4 \frac{1}{a\, \cos \eta } \frac{\partial}{\partial \phi_2} +\frac{3a'}{2a}\gamma^1+\frac{\cot(2\eta)}{a}\gamma^2. 
\end{align*}
Thus the pseudodifferential symbol of $D$ is given by
\begin{eqnarray} \nonumber
\sigma_D({ x,\xi}) = i\xi_1\gamma^1+ \frac{i\xi_2}{a}\gamma^2+ \frac{i\xi_3}{a\, \sin \eta}\gamma^3+ \frac{i\xi_4}{a\, \cos \eta } \gamma^4 +\frac{3a'}{2a}\gamma^1+\frac{\cot(2\eta)}{a}\gamma^2. 
\end{eqnarray}

For the purpose of employing pseudodifferential calculus in the 
sequel to compute the heat coefficients, we record in the following proposition the 
pseudodifferential symbol of $D^2$. This can be achieved by 
a straightforward computation to find an explicit expression 
for $D^2$, or alternatively, one can apply the composition 
rule for symbols, 
$\sigma_{P_1 P_2}({ x,\xi})=\sum_\alpha \frac{(-i)^{|\alpha|}}{\alpha !}\partial^\alpha_\xi\sigma_{P_1}\partial^\alpha_{ x}\sigma_{P_2}$,  to the symbol of $D$.

\begin{proposition}
The pseudodifferential symbol of $D^2$, where $D$ is the Dirac operator for the 
Robertson-Walker metric, is given by 
\[
\sigma(D^2)= p_2 + p_1 + p_0, 
\]
where the homogeneous components $p_i$ of order $i$ are written as 
\begin{eqnarray}\label{symbolHopf}
p_2&=&\xi _1^2+\frac{1}{a^2}\xi _2^2+\frac{1}{a^2  \sin ^2(\eta )}\xi _3^2+\frac{ 1}{a^2\cos ^2(\eta)}\xi _4^2, \nonumber \\ 
 p_1&=& \frac{-3 i a a'}{a^2}\xi _1+\frac{-i  a' \gamma ^{12}-2i \cot (2\eta )}{a^2}\xi _2 
 -\frac{i  a' \csc (\eta ) \gamma ^{13}+i \cot (\eta ) \csc (\eta ) \gamma ^{23}}{a^2}\xi _3 \nonumber \\ &&
 +\frac{i \tan (\eta ) \sec (\eta ) \gamma ^{24}-i a' \sec (\eta ) \gamma ^{14} }{a^2}\xi _4 , \nonumber \\
p_0&=&\frac{1}{4 a(t)^2}\Big(-6 a(t) a''(t)-3 a'(t)^2+\csc ^2(\eta )+\sec ^2(\eta ) \nonumber \\ 
&&+4+2  a'(t) (\cot (\eta )-\tan (\eta ))\gamma ^{12}\Big). 
\end{eqnarray}
\end{proposition}

\section{Terms up to $a_{10}$ and their Agreement with Chamseddine-Connes' Result} \label{Termsupto10}

The computation of the terms in the expansion of the spectral action for a spin manifold, or 
equivalently the calculation of the heat coefficients, can be 
achieved by recursive formulas while working in the heat kernel scheme of local invariants 
of elliptic differential operators and index theory \cite{GilBook1}.  Pseudodifferential calculus 
is an effective tool for dealing with the necessary approximations for deriving the small time 
asymptotic expansions in which the heat coefficients appear. Universal formulas in terms of the Riemann curvature operator and its contractions and covariant 
derivatives are written in the literature only  for the terms up to $a_{10}$, namely Gilkey's formulas up to $a_6$
 \cite{GilBook1, GilBook2} and the formulas in \cite{AmsBerOc, Avr, Van} for $a_8$ and $a_{10}$.

\subsection{Small time heat kernel expansions using pseudodifferential calculus.}\label{heatcoefsbypesudo}

In \cite{GilBook1}, by appealing to the Cauchy integral formula and using pseudodifferential calculus, 
recursive formulas for the heat coefficients of elliptic differential operators are derived. That is, one 
writes  \footnote{Hereafter in this paper $t$ denotes the first variable of the space when it appears in $a(t)$ and its derivatives 
and it denotes the time when it appears in the heat operator and the associated small time asymptotic expansions.}
\[
e^{-tD^2}=-\frac{1}{2\pi i}\int_\gamma e^{-t\lambda}(D^2-\lambda)^{-1}d\lambda,
\]
where the contour $\gamma$ goes around the non-negative real axis in the counterclockwise 
direction, and one uses pseudodifferential calculus to approximate $(D^2-\lambda)^{-1}$ via the homogeneous 
terms appearing in the expansion of the symbol of the parametrix of $D^2-\lambda$.  
  Although left and  
right parametrices have the same homogeneous components, for the purpose of finding recursive formulas 
for the coefficients appearing in each component, which will be explained shortly, it is more convenient for 
us to consider the right parametrix $\tilde{R}(\lambda)$. Therefore, the next task is to  compute recursively the 
homogeneous pseudodifferential symbols $r_j$ of order $-2-j$ in the expansion of  $\sigma(\tilde{R}(\lambda))$. 
Using the calculus of symbols, with the crucial nuance that $\lambda$ is considered to be of order 2, one finds that 
\[
r_0=(p_2-\lambda)^{-1}, 
\]
and for any $n>1$
\begin{align}\label{recursive1}
r_n=-r_0\sum_{\begin{array}{c}|\alpha|+j+2-k=n\\ j<n\end{array}} \frac{(-i)^{|\alpha|}}{\alpha!}d^\alpha_\xi p_k\,d_x^\alpha r_j.
\end{align}
 
We summarize the process of obtaining the heat coefficients by explaining that one then uses these homogeneous 
terms in the Cauchy integral formula to approximate the integral kernel of $e^{-t D^2}.$ Integration of  the kernel of  this operator on the diagonal
 yields a small time asymptotic expansion of the form 
 \begin{equation}\nonumber
{\rm Tr}(e^{-tD^2})\sim \sum_{n=0}^\infty \frac{t^{(n-4)/2}}{16\pi^4}\int {\rm tr}(e_n(x)) \,dvol_g \qquad (t\to 0),  
\end{equation}
where 
\begin{equation} \label{engeneralform}
e_n(x) \sqrt{\det g}=\frac{-1}{2\pi i}\int \int_\gamma e^{-\lambda}r_n(x,\xi,\lambda)\,d\lambda \,d\xi.
\end{equation}
For detailed discussions, we refer the reader to \cite{GilBook1}.

It is clear from \eqref{symbolHopf}  that cross derivatives 
of $p_2$ vanish and  $d_\xi^\alpha p_k=0$ if $|\alpha|>k$. Furthermore, 
$\frac{\partial}{\partial\phi_k}r_n=0$ for $n \geq 0$, and  the summation  \eqref{recursive1} 
is written as 
\begin{eqnarray}\label{rnshort}
r_n&=&-r_0\,p_0\,r_{n-2} -r_0\,p_1\,r_{n-1}
+ir_0\frac{\partial}{\partial\xi_1}p_1 \frac{\partial}{\partial t }r_{n-2}
+ir_0\frac{\partial}{\partial\xi_2}p_1 \frac{\partial}{\partial \eta}r_{n-2}  \nonumber \\
&&+ir_0\frac{\partial}{\partial\xi_1}p_2 \frac{\partial}{\partial t }r_{n-1} +ir_0\frac{\partial}{\partial\xi_2}p_2 \frac{\partial}{\partial \eta }r_{n-1}
+\frac{1}{2}r_0\frac{\partial^2}{\partial\xi_1^2}p_2 \frac{\partial^2}{\partial t^2 }r_{n-2} \nonumber \\
&&+\frac{1}{2}r_0\frac{\partial^2}{\partial\xi_2^2}p_2 \frac{\partial^2}{\partial \eta^2 }r_{n-2}. 
\end{eqnarray}

Using induction, we find that
\begin{equation}\label{rnja}
r_n=\sum_{ \begin{array}{c} 2j-2-|\alpha|=n \\ n/2+1 \leq j \leq 2n+1\end{array}}r_{n,j,\alpha }(x)\, r_0^j \,\xi^\alpha. 
\end{equation}
For example, one can see that for $n=0$ the only non-zero $r_{0,j,\alpha}$ is $r_{0,1,\bf{0}}=1$, and 
for $n=1$ the non-vanishing terms are  
$$r_{1,2,{\bf e}_k}=\frac{\partial p_1}{\partial \xi_k}, \qquad r_{1,3,2{\bf e}_l+{\bf e}_k}=-2ig^{kk}\frac{\partial g^{ll}}{\partial x_k},$$
where ${\bf e}_j$ denotes the $j$-th standard unit vector in $\mathbb{R}^4$.

It then follows from the equations  \eqref{engeneralform}, \eqref{rnshort} and \eqref{rnja} that 
\begin{align}\label{en}
e_n(x) \,a(t)^{3}\sin(\eta)\cos(\eta)&=\frac{-1}{2\pi i}\int_{\mathbb{R}^4}\int_\gamma e^{-t\lambda} r_n(x,\xi,\lambda)\, d\lambda \,d\xi\nonumber\\
 &=\sum r_{n,j,\alpha}(x)\int_{\mathbb{R}^4}\xi^\alpha \frac{-1}{2\pi i}\int_\gamma e^{-t\lambda}r_0^j\,d\lambda \,d\xi\\
  &=\sum \frac{c_\alpha}{(j-1)!} r_{n,j,\alpha} \,a(t)^{\alpha_2+\alpha_3+\alpha_4+3}\sin(\eta)^{\alpha_3+1}\cos(\eta)^{\alpha_4+1},\nonumber
 \end{align}
where 
$$c_\alpha=\prod_k \Gamma\left(\frac{\alpha_k+1}{2}\right)\frac{(-1)^{\alpha_k}+1}{2}.$$
It is straightforward to justify the latter using these identities:
\begin{eqnarray}
\frac{1}{2\pi i}\int_\gamma e^{-\lambda}r_0^jd\lambda&=&(-1)^{j}\frac{(-1)^{j-1}}{(j-1)!}e^{-||\xi||^2}=\frac{-1}{(j-1)!}\prod_{k=1}^4 e^{-g^{kk}\xi_k^2}, \nonumber \\ 
\int_\mathbb{R}x^ne^{-bx^2}dx&=&\frac{1}{2} \left((-1)^n+1\right) b^{-\frac{n}{2}-\frac{1}{2}} \Gamma \left(\frac{n+1}{2}\right).\nonumber
\end{eqnarray}

A key point that facilitates our calculations and the proof of our main theorem presented 
in \S \ref{proofofrationality} is the derivation of recursive formulas for the coefficients $r_{n, j, \alpha}$ as follows.  
By substitution of \eqref{rnja} into \eqref{rnshort} we find a recursive formula of the form
\begin{align}\label{rnjarec}
r_{n,j,\alpha}&= -p_0r_{n-2,j-1,\alpha}-\sum_{k}\frac{\partial p_1}{\partial\xi_k} r_{n-1,j-1,\alpha-{\bf e}_k}\nonumber\\
&\qquad +i\sum_{k}\frac{\partial p_1}{\partial \xi_k}\frac{\partial}{\partial x_k}r_{n-2,j-1,\alpha}+i(2-j)\sum_{k,l}\frac{\partial g^{ll}}{\partial x_k}\frac{\partial p_1}{\partial \xi_k}r_{n-2,j-2,\alpha-2{\bf e}_l}\nonumber\\
&\qquad+2i\sum_k g^{kk} \frac{\partial}{\partial x_k}r_{n-1,j-1,\alpha-{\bf e}_k}+i (4-2j)\sum_{k,l}g^{kk}\frac{\partial g^{ll}}{\partial x_k} r_{n-1,j-2,\alpha-2{\bf e}_l-{\bf e}_k}\\
&\qquad +\sum_k g^{kk} \frac{\partial^2}{\partial x_k^2}r_{n-2,j-1,\alpha}+(4-2j)\sum_{k,l}g^{kk}\frac{\partial g^{ll}}{\partial x_k} \frac{\partial}{\partial x_k}r_{n-2,j-2,\alpha-2{\bf e}_l}\nonumber\\
&\qquad +(2-j)\sum_{k,l}g^{kk} \frac{\partial^2 g^{ll}}{\partial x_k^2} r_{n-2,j-2,\alpha-2{\bf e}_l}\nonumber \\
 &\qquad +(3-j)(2-j)\sum_{k,l,l'}g^{kk} \frac{\partial g^{ll}}{\partial x_k}\frac{\partial g^{l'l'}}{\partial x_k} r_{n-2,j-3,\alpha-2{\bf e}_l-2{\bf e}_{l'}}.\nonumber
\end{align}

It is undeniable that the mechanism described above  for computing the 
heat coefficients involves heavy computations which need to be overcome by computer 
programming. Calculating explicitly the functions $e_n(x)$, $n=0, 2, \dots, 12$, and computing 
their integrals over $\mathbb{S}_a^3$ with computer assistance, we find the explicit polynomials in $a(t)$ and its 
derivatives recorded in the sequel, which describe the corresponding  terms in the expansion of the spectral action for the Robertson-Walker metric. That is, each function $a_n$ recorded below is the outcome of 
\begin{eqnarray}
a_n &=&\frac{1}{16\pi^4}\int_{\mathbb{S}_a^3}{\rm tr}(e_n) \,dvol_g \nonumber \\
&=&\frac{1}{16\pi^4}\int_0^{2\pi}\int_0^{2\pi}\int_0^{\pi/2}{\rm tr}(e_n) \, a(t)^{3}\sin(\eta)\cos(\eta) \,d\eta \,d\phi_1 
\,d\phi_2. \nonumber
\end{eqnarray}

\subsection{The terms up to $a_{6}$}

These terms were computed in \cite{ChaConRW} by their direct method, which is 
based on the Euler-Maclaurin summation formula and the Feynman-Kac formula, and they were checked by 
Gilkey's universal formulas. Our computations based on the method explained in the previous 
subsection also gives the same result.

The  first term, whose integral up to a universal factor gives the volume, is given by
\[
a_0=\frac{a(t)^3}{2}.  
\]
Since the latter appears as the leading term in the small time asymptotic expansion of the heat kernel 
it is related to Weyl's law, which reads the volume from the asymptotic distribution of the eigenvalues of 
$D^2$. The next term, which is related to the scalar curvature, has the expression 
\[
a_2= \frac{1}{4} a(t) \left(a(t) a''(t)+a'(t)^2-1\right).
\]
The term after, whose integral is topological, is related to the Gauss-Bonnet term  (cf. \cite{ChaConRW}) 
and is written as 
\[
a_4=\frac{1}{120} \Big(3 a^{(4)}(t) a(t)^2+3 a(t) a''(t)^2-5 a''(t)+9 a^{(3)}(t) a(t) a'(t)-4 a'(t)^2 a''(t)\Big). 
\]
The term $a_6$, which is the last term for which Gilkey's universal formulas are written, is given by 
\begin{flushleft}
$a_6=\frac{1}{5040 a(t)^2}\Big(9 a^{(6)}(t) a(t)^4-21 a^{(4)}(t) a(t)^2-3 a^{(3)}(t)^2 a(t)^3-56 a(t)^2 a''(t)^3+42 a(t) a''(t)^2+36 a^{(5)}(t) a(t)^3 a'(t)+6 a^{(4)}(t) a(t)^3 a''(t)-42 a^{(4)}(t) a(t)^2 a'(t)^2+60 a^{(3)}(t) a(t) a'(t)^3+21 a^{(3)}(t) a(t) a'(t)+240 a(t) a'(t)^2 a''(t)^2-60 a'(t)^4 a''(t)-21 a'(t)^2 a''(t)-252 a^{(3)}(t) a(t)^2 a'(t) a''(t)\Big).$ \\
\end{flushleft}

\subsection{The terms $a_8$ and $a_{10}$}

These terms were computed by Chamseddine and Connes in \cite{ChaConRW} using 
their direct method. In order to form a check on the final formulas, they have suggested  
to use the universal formulas of \cite{AmsBerOc, Avr, Van} to calculate these terms and 
compare the results. As mentioned earlier, Gilkey's universal formulas were used in \cite{ChaConRW} 
to check the terms up to $a_6$, however, they are written in the literature only up to $a_6$ and become 
rather complicated even for this term.

In this subsection, we pursue the computation of the terms $a_8$ and $a_{10}$ in the expansion of the 
spectral action for Robertson-Walker metrics by continuing 
to employ pseudodifferential calculus, as presented in \S \ref{heatcoefsbypesudo}, and check that the final 
formulas agree with the result in \cite{ChaConRW}.  The final formulas for $a_8$ and $a_{10}$ are the following expressions:

\[a_8=\]
\begin{flushleft}
\begin{small}
$-\frac{1}{10080 a(t)^4}\Big(-a^{(8)}(t) a(t)^6+3 a^{(6)}(t) a(t)^4+13 a^{(4)}(t)^2 a(t)^5-24 a^{(3)}(t)^2 a(t)^3-114 a(t)^3 a''(t)^4+43 a(t)^2 a''(t)^3-5 a^{(7)}(t) a(t)^5 a'(t)+2 a^{(6)}(t) a(t)^5 a''(t)+9 a^{(6)}(t) a(t)^4 a'(t)^2+16 a^{(3)}(t) a^{(5)}(t) a(t)^5-24 a^{(5)}(t) a(t)^3 a'(t)^3-6 a^{(5)}(t) a(t)^3 a'(t)+69 a^{(4)}(t) a(t)^4 a''(t)^2-36 a^{(4)}(t) a(t)^3 a''(t)+60 a^{(4)}(t) a(t)^2 a'(t)^4+15 a^{(4)}(t) a(t)^2 a'(t)^2+90 a^{(3)}(t)^2 a(t)^4 a''(t)-216 a^{(3)}(t)^2 a(t)^3 a'(t)^2-108 a^{(3)}(t) a(t) a'(t)^5-27 a^{(3)}(t) a(t) a'(t)^3+801 a(t)^2 a'(t)^2 a''(t)^3-588 a(t) a'(t)^4 a''(t)^2-87 a(t) a'(t)^2 a''(t)^2+108 a'(t)^6 a''(t)+27 a'(t)^4 a''(t)+78 a^{(5)}(t) a(t)^4 a'(t) a''(t)+132 a^{(3)}(t) a^{(4)}(t) a(t)^4 a'(t)-312 a^{(4)}(t) a(t)^3 a'(t)^2 a''(t)-819 a^{(3)}(t) a(t)^3 a'(t) a''(t)^2+768 a^{(3)}(t) a(t)^2 a'(t)^3 a''(t)+102 a^{(3)}(t) a(t)^2 a'(t) a''(t)\Big),$
\end{small} \\ 
\end{flushleft}
and
\[a_{10}=\]
\begin{flushleft}
\begin{small}
$\frac{1}{665280 a(t)^6}\Big(3 a^{(10)}(t) a(t)^8-222 a^{(5)}(t)^2 a(t)^7-348 a^{(4)}(t) a^{(6)}(t) a(t)^7-147 a^{(3)}(t) a^{(7)}(t) a(t)^7-18 a''(t) a^{(8)}(t) a(t)^7+18 a'(t) a^{(9)}(t) a(t)^7-482 a''(t) a^{(4)}(t)^2 a(t)^6-331 a^{(3)}(t)^2 a^{(4)}(t) a(t)^6-1110 a''(t) a^{(3)}(t) a^{(5)}(t) a(t)^6-1556 a'(t) a^{(4)}(t) a^{(5)}(t) a(t)^6-448 a''(t)^2 a^{(6)}(t) a(t)^6-1074 a'(t) a^{(3)}(t) a^{(6)}(t) a(t)^6-476 a'(t) a''(t) a^{(7)}(t) a(t)^6-43 a'(t)^2 a^{(8)}(t) a(t)^6-11 a^{(8)}(t) a(t)^6+8943 a'(t) a^{(3)}(t)^3 a(t)^5+21846 a''(t)^2 a^{(3)}(t)^2 a(t)^5+4092 a'(t)^2 a^{(4)}(t)^2 a(t)^5+396 a^{(4)}(t)^2 a(t)^5+10560 a''(t)^3 a^{(4)}(t) a(t)^5+39402 a'(t) a''(t) a^{(3)}(t) a^{(4)}(t) a(t)^5+11352 a'(t) a''(t)^2 a^{(5)}(t) a(t)^5+6336 a'(t)^2 a^{(3)}(t) a^{(5)}(t) a(t)^5+594 a^{(3)}(t) a^{(5)}(t) a(t)^5+2904 a'(t)^2 a''(t) a^{(6)}(t) a(t)^5+264 a''(t) a^{(6)}(t) a(t)^5+165 a'(t)^3 a^{(7)}(t) a(t)^5+33 a'(t) a^{(7)}(t) a(t)^5-10338 a''(t)^5 a(t)^4-95919 a'(t)^2 a''(t) a^{(3)}(t)^2 a(t)^4-3729 a''(t) a^{(3)}(t)^2 a(t)^4-117600 a'(t) a''(t)^3 a^{(3)}(t) a(t)^4-68664 a'(t)^2 a''(t)^2 a^{(4)}(t) a(t)^4-2772 a''(t)^2 a^{(4)}(t) a(t)^4-23976 a'(t)^3 a^{(3)}(t) a^{(4)}(t) a(t)^4-2640 a'(t) a^{(3)}(t) a^{(4)}(t) a(t)^4-12762 a'(t)^3 a''(t) a^{(5)}(t) a(t)^4-1386 a'(t) a''(t) a^{(5)}(t) a(t)^4-651 a'(t)^4 a^{(6)}(t) a(t)^4-132 a'(t)^2 a^{(6)}(t) a(t)^4+111378 a'(t)^2 a''(t)^4 a(t)^3+2354 a''(t)^4 a(t)^3+31344 a'(t)^4 a^{(3)}(t)^2 a(t)^3+3729 a'(t)^2 a^{(3)}(t)^2 a(t)^3+236706 a'(t)^3 a''(t)^2 a^{(3)}(t) a(t)^3+13926 a'(t) a''(t)^2 a^{(3)}(t) a(t)^3+43320 a'(t)^4 a''(t) a^{(4)}(t) a(t)^3+5214 a'(t)^2 a''(t) a^{(4)}(t) a(t)^3+2238 a'(t)^5 a^{(5)}(t) a(t)^3+462 a'(t)^3 a^{(5)}(t) a(t)^3-162162 a'(t)^4 a''(t)^3 a(t)^2-11880 a'(t)^2 a''(t)^3 a(t)^2-103884 a'(t)^5 a''(t) a^{(3)}(t) a(t)^2-13332 a'(t)^3 a''(t) a^{(3)}(t) a(t)^2-6138 a'(t)^6 a^{(4)}(t) a(t)^2-1287 a'(t)^4 a^{(4)}(t) a(t)^2+76440 a'(t)^6 a''(t)^2 a(t)+10428 a'(t)^4 a''(t)^2 a(t)+11700 a'(t)^7 a^{(3)}(t) a(t)+2475 a'(t)^5 a^{(3)}(t) a(t)-11700 a'(t)^8 a''(t)-2475 a'(t)^6 a''(t)\Big).$
\end{small}
\end{flushleft}

\section{Computation of the Term $a_{12}$ in the Expansion of the Spectral Action}\label{Term12}

We pursue the computation of the term $a_{12}$ in the expansion of the spectral action 
for Robertson-Walker metrics by employing pseudodifferential calculus to find 
the term $r_{12}$ for the parametrix of $\lambda - D^2$, which is homogeneous of order $-14$, 
and by performing the appropriate integrations. Since there is no universal formula in the 
literature for this term, we have performed two heavy computations, one in Hopf coordinates and the 
other in spherical coordinates, to form a check on the validity of the outcome of our calculations. 
Another efficient way of computing the term $a_{12}$ is to use the direct method of \cite{ChaConRW}.

\subsection{The result of the computation in Hopf coordinates.} \label{exprfora12} Continuing the recursive procedure 
commenced in the previous section and exploiting computer assistance, while the calculation 
becomes significantly heavier for the term $a_{12}$, we find the following expression: 
\[a_{12}=\]
\begin{flushleft}
\begin{small}
$\frac{1}{17297280 a(t)^{8}}\Big(3 a^{(12)}(t) a(t)^{10}-1057 a^{(6)}(t)^2 a(t)^9-1747 a^{(5)}(t) a^{(7)}(t) a(t)^9-970 a^{(4)}(t) a^{(8)}(t) a(t)^9-317 a^{(3)}(t) a^{(9)}(t) a(t)^9-34 a''(t) a^{(10)}(t) a(t)^9+21 a'(t) a^{(11)}(t) a(t)^9+5001 a^{(4)}(t)^3 a(t)^8+2419 a''(t) a^{(5)}(t)^2 a(t)^8+19174 a^{(3)}(t) a^{(4)}(t) a^{(5)}(t) a(t)^8+4086 a^{(3)}(t)^2 a^{(6)}(t) a(t)^8+2970 a''(t) a^{(4)}(t) a^{(6)}(t) a(t)^8-5520 a'(t) a^{(5)}(t) a^{(6)}(t) a(t)^8-511 a''(t) a^{(3)}(t) a^{(7)}(t) a(t)^8-4175 a'(t) a^{(4)}(t) a^{(7)}(t) a(t)^8-745 a''(t)^2 a^{(8)}(t) a(t)^8-2289 a'(t) a^{(3)}(t) a^{(8)}(t) a(t)^8-828 a'(t) a''(t) a^{(9)}(t) a(t)^8-62 a'(t)^2 a^{(10)}(t) a(t)^8-13 a^{(10)}(t) a(t)^8+45480 a^{(3)}(t)^4 a(t)^7+152962 a''(t)^2 a^{(4)}(t)^2 a(t)^7+203971 a'(t) a^{(3)}(t) a^{(4)}(t)^2 a(t)^7+21369 a'(t)^2 a^{(5)}(t)^2 a(t)^7+1885 a^{(5)}(t)^2 a(t)^7+410230 a''(t) a^{(3)}(t)^2 a^{(4)}(t) a(t)^7+163832 a'(t) a^{(3)}(t)^2 a^{(5)}(t) a(t)^7+250584 a''(t)^2 a^{(3)}(t) a^{(5)}(t) a(t)^7+244006 a'(t) a''(t) a^{(4)}(t) a^{(5)}(t) a(t)^7+42440 a''(t)^3 a^{(6)}(t) a(t)^7+163390 a'(t) a''(t) a^{(3)}(t) a^{(6)}(t) a(t)^7+35550 a'(t)^2 a^{(4)}(t) a^{(6)}(t) a(t)^7+3094 a^{(4)}(t) a^{(6)}(t) a(t)^7+34351 a'(t) a''(t)^2 a^{(7)}(t) a(t)^7+19733 a'(t)^2 a^{(3)}(t) a^{(7)}(t) a(t)^7+1625 a^{(3)}(t) a^{(7)}(t) a(t)^7+6784 a'(t)^2 a''(t) a^{(8)}(t) a(t)^7+520 a''(t) a^{(8)}(t) a(t)^7+308 a'(t)^3 a^{(9)}(t) a(t)^7+52 a'(t) a^{(9)}(t) a(t)^7-2056720 a'(t) a''(t) a^{(3)}(t)^3 a(t)^6-1790580 a''(t)^3 a^{(3)}(t)^2 a(t)^6-900272 a'(t)^2 a''(t) a^{(4)}(t)^2 a(t)^6-31889 a''(t) a^{(4)}(t)^2 a(t)^6-643407 a''(t)^4 a^{(4)}(t) a(t)^6-1251548 a'(t)^2 a^{(3)}(t)^2 a^{(4)}(t) a(t)^6-43758 a^{(3)}(t)^2 a^{(4)}(t) a(t)^6-4452042 a'(t) a''(t)^2 a^{(3)}(t) a^{(4)}(t) a(t)^6-836214 a'(t) a''(t)^3 a^{(5)}(t) a(t)^6-1400104 a'(t)^2 a''(t) a^{(3)}(t) a^{(5)}(t) a(t)^6-48620 a''(t) a^{(3)}(t) a^{(5)}(t) a(t)^6-181966 a'(t)^3 a^{(4)}(t) a^{(5)}(t) a(t)^6-18018 a'(t) a^{(4)}(t) a^{(5)}(t) a(t)^6-319996 a'(t)^2 a''(t)^2 a^{(6)}(t) a(t)^6-11011 a''(t)^2 a^{(6)}(t) a(t)^6-115062 a'(t)^3 a^{(3)}(t) a^{(6)}(t) a(t)^6-11154 a'(t) a^{(3)}(t) a^{(6)}(t) a(t)^6-42764 a'(t)^3 a''(t) a^{(7)}(t) a(t)^6-4004 a'(t) a''(t) a^{(7)}(t) a(t)^6-1649 a'(t)^4 a^{(8)}(t) a(t)^6-286 a'(t)^2 a^{(8)}(t) a(t)^6+460769 a''(t)^6 a(t)^5+1661518 a'(t)^3 a^{(3)}(t)^3 a(t)^5+83486 a'(t) a^{(3)}(t)^3 a(t)^5+13383328 a'(t)^2 a''(t)^2 a^{(3)}(t)^2 a(t)^5+222092 a''(t)^2 a^{(3)}(t)^2 a(t)^5+342883 a'(t)^4 a^{(4)}(t)^2 a(t)^5+36218 a'(t)^2 a^{(4)}(t)^2 a(t)^5+7922361 a'(t) a''(t)^4 a^{(3)}(t) a(t)^5+6367314 a'(t)^2 a''(t)^3 a^{(4)}(t) a(t)^5+109330 a''(t)^3 a^{(4)}(t) a(t)^5+7065862 a'(t)^3 a''(t) a^{(3)}(t) a^{(4)}(t) a(t)^5+360386 a'(t) a''(t) a^{(3)}(t) a^{(4)}(t) a(t)^5+1918386 a'(t)^3 a''(t)^2 a^{(5)}(t) a(t)^5+98592 a'(t) a''(t)^2 a^{(5)}(t) a(t)^5+524802 a'(t)^4 a^{(3)}(t) a^{(5)}(t) a(t)^5+55146 a'(t)^2 a^{(3)}(t) a^{(5)}(t) a(t)^5+226014 a'(t)^4 a''(t) a^{(6)}(t) a(t)^5+23712 a'(t)^2 a''(t) a^{(6)}(t) a(t)^5+8283 a'(t)^5 a^{(7)}(t) a(t)^5+1482 a'(t)^3 a^{(7)}(t) a(t)^5-7346958 a'(t)^2 a''(t)^5 a(t)^4-72761 a''(t)^5 a(t)^4-11745252 a'(t)^4 a''(t) a^{(3)}(t)^2 a(t)^4-725712 a'(t)^2 a''(t) a^{(3)}(t)^2 a(t)^4-27707028 a'(t)^3 a''(t)^3 a^{(3)}(t) a(t)^4-819520 a'(t) a''(t)^3 a^{(3)}(t) a(t)^4-8247105 a'(t)^4 a''(t)^2 a^{(4)}(t) a(t)^4-520260 a'(t)^2 a''(t)^2 a^{(4)}(t) a(t)^4-1848228 a'(t)^5 a^{(3)}(t) a^{(4)}(t) a(t)^4-205296 a'(t)^3 a^{(3)}(t) a^{(4)}(t) a(t)^4-973482 a'(t)^5 a''(t) a^{(5)}(t) a(t)^4-110136 a'(t)^3 a''(t) a^{(5)}(t) a(t)^4-36723 a'(t)^6 a^{(6)}(t) a(t)^4-6747 a'(t)^4 a^{(6)}(t) a(t)^4+17816751 a'(t)^4 a''(t)^4 a(t)^3+721058 a'(t)^2 a''(t)^4 a(t)^3+2352624 a'(t)^6 a^{(3)}(t)^2 a(t)^3+274170 a'(t)^4 a^{(3)}(t)^2 a(t)^3+24583191 a'(t)^5 a''(t)^2 a^{(3)}(t) a(t)^3+1771146 a'(t)^3 a''(t)^2 a^{(3)}(t) a(t)^3+3256248 a'(t)^6 a''(t) a^{(4)}(t) a(t)^3+389376 a'(t)^4 a''(t) a^{(4)}(t) a(t)^3+135300 a'(t)^7 a^{(5)}(t) a(t)^3+25350 a'(t)^5 a^{(5)}(t) a(t)^3-15430357 a'(t)^6 a''(t)^3 a(t)^2-1252745 a'(t)^4 a''(t)^3 a(t)^2-7747848 a'(t)^7 a''(t) a^{(3)}(t) a(t)^2-967590 a'(t)^5 a''(t) a^{(3)}(t) a(t)^2-385200 a'(t)^8 a^{(4)}(t) a(t)^2-73125 a'(t)^6 a^{(4)}(t) a(t)^2+5645124 a'(t)^8 a''(t)^2 a(t)+741195 a'(t)^6 a''(t)^2 a(t)+749700 a'(t)^9 a^{(3)}(t) a(t)+143325 a'(t)^7 a^{(3)}(t) a(t)-749700 a'(t)^{10} a''(t)-143325 a'(t)^8 a''(t))\Big).$
\end{small}
\end{flushleft}

\subsection{Agreement of the result with computations in spherical coordinates.}

Taking a similar route as in \S \ref{DiracinHopf}, we explicitly write the Dirac 
operator for the Roberson-Walker metric in spherical coordinates  
\begin{equation} \nonumber
ds^{2}=dt^{2}+a^{2}\left(  t\right)   \big (  d\chi^{2}+\sin^{2}(\chi) \left(
d\theta^{2}+\sin^{2}(\theta) \, d\varphi^{2}\right)  \big ). 
\end{equation} 
Using the computations carried out in \cite{ChaConRW} with the orthonormal coframe 
\[ dt,  \qquad a(t)\, d \chi,  \qquad a(t)\, \sin \chi \,d \theta, \qquad a(t)\, \sin \chi  \, \sin \theta \,d \varphi,  
\]
the corresponding matrix of connection 1-forms for the Levi-Civita connection is written as 
\begin{small}
\[ \left (
\begin{array}{cccc}
0 &-a'(t)d\chi  & -a'(t)\sin(\chi)d\theta & -a'(t)\sin(\chi)\sin(\theta)d\varphi\\ 
a'(t)d\chi  &0  & -\cos(\chi)d\theta  &-\cos(\chi)\sin(\theta)d\varphi \\ 
a'(t)\sin(\chi)d\theta & \cos(\chi)d\theta & 0 &-\cos(\theta)d\varphi \\ 
 a'(t)\sin(\chi)\sin(\theta)d\varphi&\cos(\chi)\sin(\theta)d\varphi  & \cos(\theta)d\varphi & 0\\
\end{array} \right ).
\]
\end{small}

Lifting to the spin bundle by means of the Lie algebra isomorphism 
$\mu:\mathfrak{so}(4)\to \mathfrak{spin}(4)$ and writing the formula for the Dirac operator 
yield the following expression for this operator expressed in spherical coordiantes:  
\begin{eqnarray}
D &=& \gamma^1 \frac{\partial}{\partial t}+\gamma^2 \frac{1}{a}\frac{\partial}{\partial \chi}+\gamma^3 \frac{1}{a\, \sin \chi} \frac{\partial}{\partial \theta}+\gamma^4 \frac{1}{a\, \sin \chi  \, \sin \theta} \frac{\partial}{\partial \varphi}  \nonumber \\ 
&&+\frac{3a'}{2a}\gamma^1+\frac{\cot(\chi)}{a}\gamma^2+\frac{\cot(\theta)}{2a\sin(\chi)}\gamma^3. \nonumber
\end{eqnarray} 
Thus the pseudodifferential symbol of $D$ is given by 
\begin{eqnarray}
\sigma_D({ x,\xi})&=&i\gamma^1\xi_1 +\frac{i}{a}\gamma^2\xi_2+\frac{i}{a\sin(\chi)}\gamma^3\xi_3+\frac{i}{a\sin(\chi)\sin(\theta)}\gamma^4\xi_4 \nonumber \\
&&+\frac{3a'}{2a}\gamma^1+\frac{\cot(\chi)}{a}\gamma^2+\frac{\cot(\theta)}{2a\sin(\chi)}\gamma^3. \nonumber
\end{eqnarray} 
Accordingly, the symbol of $D^2$ is the sum $p_2'+p_1'+p_0'$ of three homogeneous components  
\begin{eqnarray}p_2'&=&\xi _1^2+\frac{1}{a(t)^2}\xi _2^2+\frac{1}{a(t)^2  \sin ^2(\chi )}\xi _3^2+\frac{ 1}{a(t)^2\sin ^2(\theta ) \sin ^2(\chi )}\xi _4^2, \nonumber \\
 p_1'&=&-\frac{ 3 i a'(t)}{a(t)}\xi _1-\frac{i }{a(t)^2} \left(\gamma ^{12} a'(t)+2 \cot (\chi )\right)\xi _2 \nonumber \\
&&-\frac{i }{a(t)^2} \left(\gamma ^{13} \csc (\chi ) a'(t)+\cot (\theta ) \csc ^2(\chi )+\gamma ^{23} \cot (\chi ) \csc (\chi )\right)\xi _3\nonumber \\
&&-\frac{i }{a(t)^2} (\csc (\theta )  \csc (\chi ) a'(t)\gamma ^{14}+\cot (\theta ) \csc (\theta )  \csc ^2(\chi )\gamma ^{34} \nonumber \\
&&+\csc (\theta )  \cot (\chi ) \csc (\chi )\gamma ^{24} )\xi _4, \nonumber 
\\ p_0'&=&\frac{1}{8 a(t)^2}\left(-12 a(t) a''(t)-6 a'(t)^2+3 \csc ^2(\theta ) \csc ^2(\chi )-\cot ^2(\theta ) \csc ^2(\chi )+\right. \nonumber \\ && \left. 4 i \cot (\theta ) \cot (\chi ) \csc (\chi )-4 i \cot (\theta ) \cot (\chi ) \csc (\chi )-4 \cot ^2(\chi )+5 \csc ^2(\chi )+4\right) \nonumber \\
&&-\frac{\left(\cot (\theta ) \csc (\chi ) a'(t)\right)}{2 a(t)^2}\gamma ^{13} -\frac{ \left(\cot (\chi ) a'(t)\right)}{a(t)^2}\gamma ^{12}-\frac{ (\cot (\theta ) \cot (\chi ) \csc (\chi ))}{2 a(t)^2}\gamma ^{23}. \nonumber
\end{eqnarray}

We have performed the computation of the heat coefficients up to the term $a_{12}$ using the 
latter symbols and have checked the agreement 
of the result with the computations in Hopf coordinates, presented in the previous subsections. 
This is in particular of great importance for 
the term $a_{12}$, since it ensures the validity of our computations performed in two different coordinates.

\subsection{Agreement with the full expansion for the round metric.}

We first recall the full expansion for the spectral action for the round metric, namely the case 
$a(t) = \sin (t)$, worked out in \cite{ChaConRW}. Then we show that 
the term $a_{12}$ presented in \S \ref{exprfora12} reduces correctly to the round case. 

The method devised 
in \cite{ChaConRW} has wide applicability in the spectral action computations since it can be used for the cases when 
 the eigenvalues of the square of the 
Dirac operator have a polynomial expression while their  multiplicities are also given by polynomials. 
In the case of the round metric on $\mathbb{S}^4$,  after remarkable computations based on the 
Euler-Maclaurin formula, this method leads to the following expression with control over the remainder term \cite{ChaConRW}: 
\begin{eqnarray}
   \frac{3}{4}{\rm Trace}(f(tD^2)) &=& \int_0^\infty f(tx^2)(x^3-x)dx+\frac{11 f(0)}{120}-\frac{31 f'(0) t}{2520}
+\frac{41 f''(0) t^2}{10080} \nonumber \\ 
&&-\frac{31 f^{(3)}(0) t^3}{15840}+\frac{10331 f^{(4)}(0) t^4}{8648640}-\frac{3421 f^{(5)}(0) t^5}{3931200}+\dots +R_m. \nonumber
\end{eqnarray}
This implies that the term $a_{12}$ in the expansion of the spectral action for the round metric is equal to 
$\frac{10331}{6486480}$. To check our calculations against this result, we find that for $a(t)=\sin(t)$ 
the expression for $a_{12}(t)$ reduces to  $\frac{10331 \sin ^3(t)}{8648640},$ and hence
$$a_{12}=\int_0^\pi a_{12}(\mathbb{S}^4)\,dt=\frac{4}{3} \frac{10331}{8648640}=\frac{10331}{6486480},$$ 
which is in complete agreement with the result in  \cite{ChaConRW}, mentioned above.

\section{Chameseddine-Connes' Conjecture} \label{ProofofConjecture}

In this section we prove a   conjecture of Chamseddine and Connes from  \cite{ChaConRW}.  More precisely, we   
show that the term $a_{2n}$ in the asymptotic expansion of the spectral action for Robertson-Walker 
metrics is, up to multiplication by $a(t)^{3-2n}$, of the form $Q_{2n}(a,a',\dots,a^{(2n)})$, 
where $Q_{2n}$ is a polynomial with rational coefficients. 

\subsection{Proof of rationality of the coefficients in the expressions for $a_{2n}$}\label{proofofrationality}

A crucial point that enables us to furnish the proof of our main theorem, 
namely the proof of the conjecture mentioned above, is the independence of the integral kernel of 
the heat operator of the Dirac operator 
of the Robertson-Walker metric  from the variables $\phi_1, \phi_2, \eta$. Note that since the symbol and the metric are independent of $\phi_1, \phi_2$, the computations involved in the symbol calculus clearly  imply  the independence of the terms $e_n$ from these variables. However, the independence of $e_n$ from $\eta$ is not evident, which is proved as follows.

\begin{lemma}\label{ind}
The heat kernel $k(t, x, x)$ for the Robertson-Walker metric is independent of $\phi_1,\phi_2, \eta$.
\end{lemma}
\begin{proof}
The round metric on $\mathbb{S}^3$ is the bi-invariant metric on ${\rm SU}(2)$ induced from the Killing form of its Lie algebra $\mathfrak{su}(2)$. The corresponding Levi-Civita connection restricted to the  left invariant vector fields is given by $\frac{1}{2}[X,Y]$, and to the right invariant vector fields by $\frac{-1}{2}[X,Y]$. Since the Killing form is ${\rm ad}$-invariant, we have
$$\langle [X,Y],Z\rangle+\langle Y,[X,Z]\rangle=0,\qquad X,Y,Z\in \mathfrak{su}(2),$$
which implies that in terms of the connection on left (right) invariant  vector fields $X,Y,Z$, it can be written as
\begin{equation}\label{Killingequ}
\langle \nabla_YX,Z\rangle+\langle Y,\nabla_ZX\rangle=0.
\end{equation}
Considering the fact that $\nabla X:\mathfrak{X}(M)\to \mathfrak{X}(M)$ is an endomorphism of the tangent bundle, the latter identity holds for any  $Y,Z\in\mathfrak{X}(M)$. Therefore, the equation \eqref{Killingequ} is the Killing equation and shows that any left and right invariant vector field on ${\rm SU}(2)$ is a Killing vector field.

By direct computation in  Hopf coordinates, we find the following vector fields which respectively form bases for left and right invariant vector fields on ${\rm SU}(2)$:
\begin{align*} 
X^L_1&=\frac{\partial}{\partial \phi_1}+\frac{\partial}{\partial \phi_2},\\
X^L_2&=\sin (\phi_1+\phi_2)\frac{\partial}{\partial \eta}+\cot(\eta )\cos (\phi_1+\phi_2)\frac{\partial}{\partial \phi_1}-\tan (\eta )\cos (\phi_1+\phi_2) \frac{\partial}{\partial \phi_2},\\
X^L_3&=\cos (\phi_1+\phi_2)\frac{\partial}{\partial \eta}-\cot (\eta ) \sin (\phi_1+\phi_2)\frac{\partial}{\partial \phi_1}+\tan (\eta ) \sin(\phi_1+\phi_2)\frac{\partial}{\partial \phi_2},\\ 
X^R_1&=-\frac{\partial}{\partial \phi_1}+\frac{\partial}{\partial \phi_2},\\
X^R_2&=-\sin(\phi_1-\phi_2)\frac{\partial}{\partial \eta}-\cot (\eta ) \cos (\phi_1-\phi_2)\frac{\partial}{\partial \phi_1}-\tan(\eta )\cos(\phi_1-\phi_2)\frac{\partial}{\partial \phi_2},\\
X^R_3&=\cos(\phi_1-\phi_2)\frac{\partial}{\partial \eta}-\cot (\eta ) \sin(\phi_1-\phi_2)\frac{\partial}{\partial \phi_1}-\tan(\eta )\sin (\phi_1-\phi_2)\frac{\partial}{\partial \phi_2}.
\end{align*}
One can check that these vector fields are indeed Killing vector fields for the Robertson-Walker metrics on the four dimensional space.
Thus, for any isometry invariant function $f$ we have:
\begin{eqnarray}
&&\frac{\partial}{\partial \phi_1}f=\frac{1}{2}(X^L_1-X^R_1)f=0, \nonumber \\
&&\frac{\partial}{\partial \phi_2}f=\frac{1}{2}(X^L_1+X^R_1)f=0,\nonumber \\
&&\frac{\partial}{\partial \eta}f=(\sin(\phi_1+\phi_2)X^L_2+\cos(\phi_1+\phi_2)X^L_3)f=0. \nonumber
\end{eqnarray}
In particular, the heat kernel restricted to the diagonal, $k(t, x, x)$, is independent of $\phi_1, \phi_2,\eta$, and so are the coefficient functions $e_n$ in its asymptotic expansion. 

\end{proof}

We stress that although $e_n (x)$ is independent of $\eta,\phi_1,\phi_2$, its components denoted by 
$e_{n, j, \alpha}$ in the proof of the following theorem are not necessarily independent of these variables. 

\begin{theorem} \label{rationalitytheorem}
The term $a_{2n}$ in the expansion of the spectral action for the Robertson-Walker metric with 
cosmic scale factor $a(t)$ is of the form
\[
\frac{1}{a(t)^{2n-3}}\,Q_{2n}\left(a(t),a'(t),\dots,a^{(2n)}(t)\right),
\]
where $Q_{2n}$ is a polynomial with rational coefficients.  
\end{theorem}
\begin{proof}
Using \eqref{en} we can write
 \begin{align}\label{ensum}
 e_n&=\sum_{\begin{array}{c}2j-2-|\alpha|=n\\ n/2+1\leq j\leq 2n+1\end{array}}  c_\alpha \,e_{n,j,\alpha},
 \end{align}
where
$$e_{n,j,\alpha}=\frac{1}{(j-1)!}r_{n,j,\alpha} \,a(t)^{\alpha_2+\alpha_3+\alpha_4}\sin(\eta)^{\alpha_3}\cos(\eta)^{\alpha_4}.$$
The recursive equation \eqref{rnjarec} implies that 
\begin{equation} \label{recforenja}
e_{n,j,\alpha}=
\end{equation} 
\begin{small}
\begin{flushleft} 
$\frac {1}{(j-1)a(t)}\Big( 
    (\gamma^{14} a'(t)-\tan(\eta)\gamma^{24} )e_{n-1,j-1,\alpha-{\bf e}_4}
   +(\gamma^{13}a'(t)+\cot(\eta)\gamma^{23})e_{n-1,j-1,\alpha-{\bf e}_3}
    +(\gamma^{12}a'(t)+1((2\alpha_4-1)\tan(\eta)+(1-2\alpha_3)\cot(\eta)))e_{n-1,j-1,\alpha-{\bf e}_2}
    +4a'(t)e_{n-1,j-2,\alpha-{\bf e}_1-2{\bf e}_2}
    +4a'(t)e_{n-1,j-2,\alpha-{\bf e}_1-2{\bf e}_3} 
    + 4a'(t)e_{n-1,j-2,\alpha-{\bf e}_1-2{\bf e}_4} 
    + (-2\alpha_2-2\alpha_3-2\alpha_4 + 3 )a' (t)e_{n-1,j-1,\alpha-{\bf e}_1} 
    +2a(t)\frac{\partial} {\partial t} e_{n-1,j-1,\alpha-{\bf e}_1}
 -4\tan(\eta)e_{n-1,j-2,\alpha-{\bf e}_2-2{\bf e}_4}
  +4\cot(\eta)e_{n-1,j-2,\alpha-{\bf e}_2-2{\bf e}_3}
   + 2\frac{\partial}{\partial\eta} e_{n-1,j-1,\alpha-{\bf e}_2}\Big) \newline
 + \frac {1} {(j-1)a(t)^2}
  \Big(
  a(t)^2\frac {\partial^2} 
{\partial t^2} e_{n-2,j-1,\alpha}  
+ 4a'(t)a(t)\frac {\partial}{\partial t} e_{n-2,j-2,\alpha-2{\bf e}_2}  
+ 4 a'(t)a(t)\frac{\partial} {\partial t} e_{n-2,j-2,\alpha-2{\bf e}_3} 
+ 4a' (t) a (t)\frac{\partial} {\partial t} e_{n-2,j-2,\alpha-2{\bf e}_4} 
 + (-2\alpha_2-2\alpha_3-2\alpha_4+3)
 a'(t)a(t) 
 \frac{\partial} {\partial t} e_{n-2,j-1,\alpha}
 + 4a'(t)^2 e_{n-2,j-3,\alpha-4{\bf e}_2}  
 +8a'(t)^2 e_{n-2,j-3,\alpha-2{\bf e}_2-2{\bf e}_3}  
 +8 a' (t)^2e_{n-2,j-3,\alpha-2{\bf e}_2-2{\bf e}_4} 
    +4\cot(\eta)\frac{\partial}{\partial\eta}e_{n-2,j-2,\alpha-2{\bf e}_3}
    - 4\tan(\eta)\frac {\partial} {\partial\eta} e_ {n-2,j-2,\alpha-2{\bf e}_4} 
     +\frac{\partial^2}{\partial\eta^2} e_{n-2,j-1,\alpha} 
     +\big(2\cot(\eta)\gamma^{12}a'(t)+
         (-
         4(\alpha_2+\alpha_3+\alpha_4-2)a'(t)^2+ 4 (-(\alpha_3-1)\csc^2(\eta)+\alpha_3+\alpha_4-2) + 
         2 a(t)a''(t)
         )
     \big)e_{n-2,j-2,\alpha-2{\bf e}_3}
+ \big(
     (\cot(\eta)(1-2\alpha_3)+(2\alpha_4-1)\tan (\eta)) + \gamma^{12} a'(t)
\big)\frac {\partial}{\partial\eta} e_{n-2,j-1,\alpha}  
 +\big((-4(\alpha_2+\alpha_3+\alpha_4- 2) a'(t)^2 +4(-(\alpha_4-1)\sec^2(\eta)+\alpha_3 + \alpha_4-2 )+2a(t)a''(t)) - 2\gamma^{12}\tan(\eta)a'(t)
\big)e_{n-2,j-2,\alpha-2{\bf e}_4} 
     +8(a'(t)^2-1) 
     e_{n-2,j-3,\alpha-2{\bf e}_4}
    + 4(\cot^2(\eta)+a'(t)^2)
    e_{n-2,j-3,\alpha-4{\bf e}_3}
      +4(\tan^2(\eta)+a'(t)^2 )
      e_{n-2,j-3,\alpha-4{\bf e}_4}
       + (2a(t)a''(t)-4(\alpha_2+\alpha_3+ \alpha_4-2)a'(t)^2)e_{n-2,j-2,\alpha-2{\bf e}_2}
       +\big(\frac{1}{2}
        	(\cot(\eta)(1-2\alpha_3)+(2\alpha_4-1)\tan(\eta))\gamma^{12} a'(t)+\frac{1}{4}
        	((4\alpha_3^2-1)\csc^2 (\eta)-4(\alpha_3+\alpha_4-1)^2 +(2\alpha_2+2\alpha_3+2\alpha_4-3)(2\alpha_2+2\alpha_3+2\alpha_4-1) a'(t)^2 +\sec^2(\eta)(4\alpha_4^2-1)- 2(2\alpha_2+2\alpha_3+2\alpha_4-3)a(t) a''(t)
        	) 
        \big)e_{n-2,j-1,\alpha}
        \Big).
   $
\end{flushleft}
\end{small}
The functions associated with the initial indices are:    
\begin{eqnarray}
&&e_{0,1,0,0,0,0}=1, \qquad e_ {1,2,1,0,0,0}= \frac {3ia'(t)}{a(t)}, \qquad e_ {1,3,1,2,0,0}= \frac {2ia'(t)}{a(t)}, \nonumber \\
&& e_ {1,3,1,0,2,0}= \frac{2ia'(t)}{a (t)}, \qquad
e_ {1,3,1,0,0,2}= \frac{2ia'(t)}{a (t)}, \qquad 
e_ {1,3,0,1,0,2}=-\frac{(2 i)\tan(\eta)}{a(t)}, \nonumber \\
&&e_ {1,3,0,1,2,0}=\frac{(2 i)\cot(\eta)}{a(t)}, \qquad 
e_ {1, 2, 0, 0, 1,0}= \frac {i\gamma ^{13} a'(t)} {a (t)}+\frac{i\gamma^{23}\cot(\eta)}{a (t)}, \nonumber \\ 
&& e_ {1,2,0,0,0,1}=\frac{i\gamma ^{14}a'(t)}{a(t)}-\frac{i\gamma^{24}\tan(\eta)}{a (t)}, \qquad 
e_ {1,2,0,1,0,0}=\frac{2i\cot(2\eta)}{a (t)}+ \frac{i\gamma^{12}a'(t)}{a(t)}. \nonumber
\end{eqnarray}

It is then apparent that  $e_{0}$ and $e_1$ are, respectively, a polynomial in 
$a(t)$, and a polynomial in $a(t)$ and $a'(t)$, divided by some powers of $a(t)$. 
Thus, it follows from the above recursive formula that all  $e_{n,j,\alpha}$ are of this 
form. Accordingly, we have $$e_n=\frac{P_n}{a(t)^{d_n}},$$ where $P_n$ is a 
polynomial in $a(t)$ and its derivatives with matrix coefficients. 
Writing $e_{n,j,\alpha}=P_{n,j,\alpha}/a(t)^{d_n}$, we obtain $d_n=\max\{d_{n-1}+1,d_{n-2}+2\}.$
Starting with $d_0=0$, $d_1=-1$, and following to obtain $d_n=n$, we conclude that  
\[
e_{n,j,\alpha}=\frac{1}{a^{n}(t)}P_{n,j,\alpha}(a(t),\dots,a^{(n)}(t)),
\]
where $P_{n,j,\alpha}$ is a polynomial whose coefficients are matrices with entries in 
the algebra generated by $\sin(\eta),\csc(\eta), \cos(\eta), \sec(\eta)$ and rational numbers.

In the calculation of the even terms $a_{2n}$, only even $\alpha_k$ have contributions in
the summation \eqref{ensum}. This implies that the corresponding $c_\alpha$ is a rational multiple of $\pi^2$ 
and  $P_{2n}$ is a polynomial with rational matrix coefficients, which is independent of 
variables $\eta,\phi_1,\phi_2$ by Lemma \ref{ind}. Hence
\[
a_{2n}=\frac{1}{16\pi^4}\int_{\mathbb{S}_a^3}{\rm tr}(e_{2n})\,dvol_g=
\frac{2\pi^2a(t)^3}{16\pi^4}\,{\rm tr}\Big(\frac{P_{2n}}{a(t)^{2n}}\Big)=
\frac{Q_{2n}}{a(t)^{n-3}},
\]
where $Q_{2n}$ is a polynomial in $a(t), a'(t), \dots, a^{(2n)}(t)$ with rational coefficients. 
\end{proof}

The polynomials $P_{n, j, \alpha}$ also satisfy recursive relations that  
illuminate interesting features about their structure.

\begin{proposition} \label{monomialformofpnja}

Each $P_{n, j ,\alpha}$ is a finite sum of the form 
\[
\sum c_k \, a(t)^{k_0}a'(t)^{k_1}\cdots a^{(n)}(t)^{k_n},
\]
where each $c_k$ is a matrix of functions that are independent from the variable $t$, and 
$
\sum_{j=0}^n k_j=\sum_{j=0}^n jk_j=l,$
for some $0 \leq l \leq n$. 

\end{proposition}

 \begin{proof}

This  follows from  an algebraically lengthy recursive formula for $P_{n, j ,\alpha}$ 
which stems from the equation \eqref{rnjarec}, similar to the recursive formula for $e_{n, j, \alpha}$ in the proof  of Theorem  \ref{rationalitytheorem}. In addition, one needs to find the following initial cases: 

\begin{eqnarray}
&&P_ {0, 1, 0, 0, 0, 0} = I,  \qquad P_ {1, 2, 1, 0, 0, 0} = 3 i a' (t), \quad P_ {1, 2, 0, 0, 1, 0} = i\gamma ^{13} a' (t) + i\gamma ^{23}\cot (\eta),  \nonumber \\ 
&&P_ {1, 2, 0, 0, 0, 1} = i\gamma ^{14} a' (t) - i\gamma^{24}\tan (\eta), \qquad P_ {1, 2, 0, 1, 0, 0} = 2 i\cot (2\eta) + i\gamma^{12} a' (t), \nonumber \\
&&P_ {1, 3, 0, 1, 0, 2} = -2 i\tan (\eta),\qquad P_ {1, 3, 0, 1, 2, 0} = 2 i\cot (\eta), \qquad P_ {1, 3, 1, 2, 0, 0} = 2 i a' (t),\nonumber \\
&&P_ {1, 3, 1, 0, 2, 0} = 2 i a' (t),  \qquad P_ {1, 3, 1, 0, 0, 2} = 2 i a' (t). \nonumber 
\end{eqnarray}
\end{proof}

\subsection{A recursive formula for the coefficient of the highest order term in $a_{2n}$}

The highest derivative of the cosmic scale factor $a(t)$ in the expression for $a_n$ is seen in the term 
$a(t)^{n-1}a^{(n)}(t)$, which has a rational coefficient based on Theorem \ref{rationalitytheorem}. Let us 
denote the coefficient of  $a(t)^{n-1}a^{(n)}(t)$ in $a_n$ by $h_n$. Since the coefficients $h_n$ are limited to satisfy 
the recursive relations derived in the proof of the following proposition, one can find the following closed 
formula for these coefficients.

\begin{proposition}
The coefficient $h_n$ of $a(t)^{n-1}a^{(n)}(t)$ in $a_n$ is equal to 
\[
\sum_{\begin{array}{c}
[n/2]+1\leq j\leq 2n+1\\ 
0\leq k \leq j-n/2-1
\end{array}} \Gamma\left(\frac{2k+1}{2}\right)H_{n,j,2k},
\] 
where, starting from   
\begin{eqnarray}
&&H_{1,2,1}=H_{1,3,1}=\frac{3 i}{2 \sqrt{\pi }},\qquad
H_{2,4,2}=-\frac{1}{\sqrt{\pi }},    \nonumber \\
&&H_{2,3,0}=H_{2,2,0}=\frac{3}{4 \sqrt{\pi }},\qquad  H_{2,3,2}=-\frac{3}{2 \sqrt{\pi }}, \nonumber
\end{eqnarray}
the quantities $H_{n,j,\alpha}$ are computed recursively by 
\[H_{n,j,\alpha}=\frac{1}{j-1}(H_{n-2,j-1,\alpha}+2 i  H_{n-1,j-1,\alpha-1}).\]

\end{proposition}

\begin{proof}
It follows from Proposition \ref{monomialformofpnja} that 
the highest derivative of $a(t)$  in $a_n$  appears in the term $a(t)^{n-1}a^{(n)}(t)$. By a careful analysis 
of  the equation \eqref{recforenja} we find that  only the terms
\[
\frac{1}{j-1}\Big (a(t)^2 \frac{\partial ^2}{\partial t^2}P_{n-2,j-1,\alpha}+
2 i a(t) \frac{\partial}{\partial t} P_{n-1,j-1,\alpha-{\bf e}_1} \Big)
\]
contribute to its recursive formula. Denoting the corresponding monomial in 
$P_{n,j,\alpha}$ by $H_{n,j,\alpha}a(t)^{n-1}a^{(n)}(t)$  and substituting it into 
the above formula we obtain the equation 
\[
H_{n,j,\alpha}=\frac{1}{j-1}(H_{n-2,j-1,\alpha}+2 i  H_{n-1,j-1,\alpha-{\bf e}_1}),
\]
for any $n>2$. Denoting 
$$H_{n,j,\alpha_1}=\sum \prod_{k=2}^4 \Gamma\left(\frac{\alpha_k+1}{2}\right)\frac{(-1)^{\alpha_k}+1}{2}\,{\rm tr}\left(\frac{ 1}{(2 \pi )^2}\int _0^{\pi/2 }H_{n,j,\alpha_1,\alpha_2,\alpha_3,\alpha_4}d\eta \right),$$
the recursive formula converts to 
$$H_{n,j,\alpha}=\frac{1}{j-1}(H_{n-2,j-1,\alpha}+2 i  H_{n-1,j-1,\alpha-1}).$$
Thus,  the coefficient of $a(t)^{n-1}a^{(n)}(t)$ in $a_n$ is given by the above expression. 
\end{proof}

Using the above proposition we find that:  
\begin{eqnarray}
&& h_2 = \frac{1}{4}, 
\qquad  h_4=  \frac{1}{40}, 
\qquad h_6= \frac{1}{560}, 
\qquad h_8= \frac{1}{10080}, 
\qquad h_{10}=\frac{1}{221760}, \nonumber \\
&&  h_{12}=\frac{1}{5765760}, 
\qquad h_{14} = \frac{1}{172972800},  
\qquad  h_{16} = \frac{1}{5881075200}, \nonumber \\
&& h_{18}=\frac{1}{223480857600}, 
\qquad h_{20}= \frac{1}{9386196019200}.  \nonumber
\end{eqnarray}

\section{Conclusions} \label{Conclusions}

Pseudodifferential calculus is an effective tool for applying heat kernel methods to compute the 
terms in the expansion of a spectral action. We have used  this technique to 
derive  the terms up to $a_{12}$ in the expansion of the spectral action for 
the Robertson-Walker metric on a 4-dimensional geometry with a general cosmic scale factor $a(t)$. 
Performing the computations in Hopf coordinates, which reflects the symmetry of the space more conveniently   
at least from a technical point of view, we proved the independence of the integral kernel of the corresponding heat 
operator from three coordiantes of the space.  This allowed us to furnish the proof of the conjecture of Chamseddine 
and Connes on rationality of the coefficients of the polynomials in  $a(t)$ and its derivatives  that describe the general terms $a_{2n}$ in the expansion.

The terms up to $a_{10}$ were previously computed in \cite{ChaConRW} using their direct method, where the terms up to 
$a_6$ were checked against Gilkey's universal formulas \cite{GilBook1, GilBook2}. The outcome of our computations confirms the previously 
computed terms. Thus, we have formed a check on the terms $a_8$ and $a_{10}$. In order to confirm our calculation   
for the term $a_{12}$, we have performed a completely different computation in spherical coordinates and checked 
its agreement with our calculation in Hopf coordinates. It is worth emphasizing that  the high complexity of the computations, which is overcome by 
computer assistance, raises the need to derive the expressions at least in two different ways to ensure their validity.

We have found a formula for the coefficient of the term with the highest derivative of $a(t)$ in $a_{2n}$ for all $n$ 
and make the following observation.  The polynomials $Q_{2n}$ in 
$a_{2n}=Q_{2n}\left(a(t),a'(t),\dots,a^{(2n)}(t)\right)/a(t)^{2n-3}$  are of the following form up to $Q_{12}$: 
\[
Q_{2n}(x_0, x_1, \dots, x_{2n}) = \sum  c_k\, x_0^{k_0} x_1^{k_1} \cdots x_{2n}^{k_{2n}}, \qquad c_k \neq0,
\]
where the summation is over all tuples of non-negative integers $k=(k_0,k_1, \dots, k_{2n})$ such that either 
$ \sum  k_j   =  2n$  while $  \sum j k_j  =   2n$, or  $ \sum k_j   =   2n-2$   while   $ \sum j k_j   =    2n-2$.  This provides enough evidence and hope to shed more light on general structure of the terms $a_{2n}$ by further investigations, which are under way.

\section*{Acknowledgments}

We are indebted to Alain Connes for helpful discussions and encouragements on the present topic. 
F.F. thanks the Institut des Hautes \'Etudes Scientifiques (I.H.E.S.) and its IT department, in particular Francois Bachelier, for their support and the excellent environment and facilities during his visit in the Fall of 2013.

\end{document}